\documentclass[11pt]{article}
\usepackage{graphicx}
\usepackage[small,compact]{titlesec}
\usepackage{algorithm,algorithmic}
\usepackage{amsmath, amssymb, amsthm}
\usepackage{url}
\usepackage{fullpage, prettyref}
\usepackage{pstricks,pst-node,pst-text}
\usepackage{boxedminipage}
\usepackage{hyperref}
\usepackage{wrapfig}
\usepackage{ifthen}
\newtheorem{theorem}{Theorem}
\newtheorem{lemma}{Lemma}
\newtheorem{claim}{Claim}
\newtheorem{corollary}{Corollary}

\newcommand{\ignore}[1]{}

\def\bar{\overline}

\def\Exp{{\mathrm E}}

\def\prefers{>}
\def\M{\mathcal{M}}
\def\A{\mathcal{A}}
\def\perf{52.6\%}
\def\opt{{\tt OPT}}
\def\alg{{\tt ALG}}
        {\hspace*{\fill}$\Box$\par}

\begin{document}

\title{Social Welfare in One-sided Matching Markets without Money}
\author{Anand Bhalgat\footnote{Dept. of Comp. and Inf. Science, Univ. of Pennsylvania, Philadelphia, PA 19104. \newline Email: {\tt {bhalgat,deepc}@seas.upenn.edu, sanjeev@cis.upenn.edu}} \and  Deeparnab Chakrabarty\footnotemark[\value{footnote}] \and Sanjeev Khanna\footnotemark[\value{footnote}]}
\date{}
\maketitle
\begin{abstract}
We study social welfare in one-sided matching markets where the goal is
to efficiently allocate $n$ items to $n$ agents that each have
a complete, private preference list and a unit demand over the items. Our focus is
on allocation mechanisms that do not involve any monetary payments.
We consider two natural measures of social welfare: the {\em ordinal welfare factor} which measures the number of agents that are at least as happy as in some unknown, arbitrary benchmark allocation, and the {\em linear welfare factor} which assumes an agent's utility linearly decreases down his preference lists, and measures the total utility to that achieved by an optimal allocation. 
%A (randomized) mechanism is said to have a welfare factor of $\alpha$
%with respect to either measure if it achieves (in expectation) at least $\alpha$-fraction of the optimal value.

\vspace*{1mm}
We analyze two matching mechanisms which have been extensively studied by economists.
The first mechanism is the random serial dictatorship (RSD)
where agents are ordered in accordance with a randomly chosen permutation, and  are successively allocated their best choice among the unallocated items. The second mechanism is the probabilistic serial (PS) mechanism of Bogomolnaia and Moulin \cite{bogo-moulin},
which computes a fractional allocation that can be expressed as a convex combination of integral allocations. The welfare factor of a mechanism is the infimum over all instances.
For RSD, we show that the ordinal welfare factor is asymptotically $1/2$, while 
the linear welfare factor lies in the interval $[.526,2/3]$. For PS, we show that the 
ordinal welfare factor is also $1/2$ while the linear welfare factor is roughly $2/3$.
%An allocation generated by PS is envy-free with respect to both measures above, and is truthful with respect to the ordinal measure.
%We show that the ordinal welfare factor of PS is exactly $1/2$ while its linear welfare factor is roughly $2/3$.
%Finally, we extend our analysis to the case when the preference lists of agents may be incomplete and when the agent
%demands are non-unit. 
To our knowledge, these results are the first non-trivial performance guarantees 
%on ordinal and linear welfare factor of 
for these natural mechanisms.
\end{abstract}

\noindent
%\vfill\eject
%{\small
%{\bf Keywords:} Mechanism Design without Money; Matching Markets; Analysis of Algorithms
%}

\thispagestyle{empty}
\newpage
\setcounter{page}{1}

\section{Introduction}

In the one-sided matching market problem\footnote{In the literature, the problem has been alternately called the
{\em house allocation} or {\em assignment} problem.},
the goal is to efficiently allocate $n$ items, $I$, to $n$ unit-demand agents ,$A$, with each agent $a$ having
a complete and private preference list $\ge_a$ over these items. The problem arises in various applications such as
assigning dormitory rooms to students, time slots to users of a common machine, organ allocation markets,
and so on. Since the preferences are private, we focus on {\em truthful} (strategyproof)
mechanisms in which agents do not have an incentive to misrepresent their preferences.
One class of such mechanisms involve
monetary compensations/payments among agents. 
However, in many cases (e.g., in the examples cited above), monetary transfer may be infeasible due to reasons varying from legal restrictions to plain inconvenience. Hence, we focus on truthful mechanisms
without money.

A simple mechanism for the one-sided matching problem is the following: agents arrive one-by-one 
according to a fixed order $\sigma$ picking up their most preferred unallocated item.
This is called a {\em serial dictatorship} mechanism.
The  {\em random serial dictatorship} (RSD) mechanism picks the order $\sigma$ uniformly at random among all permutations.
Apart from being simple and easy to implement, RSD has attractive properties: it is truthful, fair, anonymous/neutral,
non-bossy\footnote{A mechanism is neutral if the allocation of items doesn't change with renaming, and is non-bossy if no agent can change his preference so that his allocation remains unchanged while someone else's changes.}, and returns a Pareto optimal allocation. In fact, it is the only truthful mechanism with the above properties~\cite{Svensson}, and there is a large body of economic literature on this mechanism
 (see Section \ref{sec:rel}).

%{\bf P3 Our Contribution.}
Despite this, an important question has been left unaddressed: how {\em efficient} is this mechanism?
To be precise, what is the guarantee one can give on the social welfare obtained by this algorithm when compared to
the optimal social welfare? As computer scientists, we find this a natural and important question, and
we address it in this paper.

How should one measure the social welfare  of a mechanism?
The usual recourse is to {\em assume} the existence of {\em cardinal} utilities $u_{ij}$ of agent $i$ for item $j$ which induces the preference list over items with the semantic that agent $i$ prefers item $j$ to $\ell$ iff $u_{ij} > u_{i\ell}$. 
A mechanism has {\em welfare factor} $\alpha$ if for every instance the utility of the matching returned is at least $\alpha$ times that of the optimum utility matching. There are a couple of issues with this.
Firstly, it is not hard to see that one cannot say anything meaningful about the performance of RSD if the utilities are allowed to be arbitrary. This is because the optimum utility matching might be arising due to one particular agent getting one particular item (a single edge), however with high probability, any random permutation would lead to another agent getting the item and lowering the total welfare by a lot \footnote{The reader may notice similarities of RSD with online algorithms for bipartite matching problems. We elaborate on the connection in Section \ref{sec:prelim-kvv}.}. Secondly, the assumption of cardinal utilities
inherently ties up the performance of the algorithm with the `cardinal numbers' involved; the very quantities whose existence is only an assumption. Rather, what is needed is an {\em ordinal} scale of analyzing the quality of a mechanism; a measure that depends only on the order/preference lists of the agents rather than the precise utility values.

In this paper, we propose such a measure which we call the {\em ordinal social welfare} of a mechanism.
Given an instance of items and agents with their preference lists, we assume that there exists some benchmark matching $M^*$, unknown to the mechanism. We stress here this can be {\em any} matching.
We say that the {\em ordinal welfare factor} of a mechanism is $\alpha$, if for any instance and every benchmark matching $M^*$, in the matching returned by the mechanism, at least $\alpha n$ agents (in expectation) get an item which they prefer at least as much as their allocation in $M^*$.

A discussion of this measure is in order.
Firstly, the measure is {\em ordinal} and is well defined whenever the utilities are expresses via preference lists. 
Secondly, the notion is independent of any `objective function' that an application might give rise to since it measures the ordinal social welfare with respect to any desired matching.
One disadvantage of the concept is that it is global: it counts the fraction of the {\em total} population which gets better than their optimal match. In other words, if everyone is ``happy'' in the benchmark matching $M^*$, then a mechanism with the ordinal welfare factor $\alpha$ will make
an  $\alpha$ fraction of the agents happy. However if $M^*$ itself is {\em inefficient}, say only $1\%$ of the agents are ``satisfied'' in $M^*$, then the ordinal welfare factor does not say much.
For instance, it does not help for measures like ``maximize number of agents getting their first choice", for in some instances, this number could be tiny even in $M^*$. Furthermore, it does not say anything about the ``happiness'' of individual agents, e.g. a mechanism may have the ordinal welfare factor close to $1$, but there may exists an agent who is almost always allocated an item that he prefers less than $M^*$.
Finally, we observe that the ordinal welfare factor of any mechanism, even ones which know the true preference lists, cannot be larger than $1/2$. The reason for this is that the allocation must be competitive with respect to all benchmark matchings simultaneously,
and it can be seen (Theorem \ref{thm:rsd-oe-lb}) that in the instance when all agents have the same preference list, if $M^*$ is chosen to be a random allocation, then no mechanism can have an ordinal welfare factor better than $1/2$.
 As our first result, we show that the ordinal welfare factor of RSD is in fact asymptotically $1/2$.
 \begin{theorem}\label{thm:rsd-oe}
 Given any instance and any matching $M^*$, the expected fraction of agents getting an
 item in RSD which is as good as what they get in $M^*$ is at least  $1/2 - o(1)$.
\end{theorem}
Till now we have focussed on the RSD mechanism since it is a simple (and almost unique) truthful mechanism for the matching market problem.
A mechanism is called truthful if misrepresenting his preference list doesn't strictly increase the
total {\em utility} of an agent, where the utility is defined as the cardinal utility obtained by the agent on getting his allocated item. However, when the utilities of agents are represented as preference lists, one needs a different definition. 
\iffalse How should truthfulness be defined if the utilities are given as preference lists?
One possibility is to call a mechanism truthful if misrepresenting doesn't give an agent a strictly better item in his preference list. While this suffices for deterministic mechanisms, this is ill-defined for randomized mechanisms (like RSD) which return a distribution on items for every agent. To make this well-defined one needs a way to compare two distributions. \fi
In light of this, Bogomolnaia and Moulin \cite{bogo-moulin} proposed a notion of truthfulness based on the {\em stochastic dominance}. Roughly speaking, for an agent a random allocation rule stochastically dominates another if the probability of getting one of his top $k$ choices in the first rule is at least that in the second, for any $k$.

Using this notion, \cite{bogo-moulin} called a mechanism 
(weakly) truthful if
no agent can obtain a stochastically dominating allocation by misreporting his preference list. 
In other words, there exists at least one utility function $u$ that maps a preference list over items to cardinal utility values, such that the mechanism is truthful if the user's utility function is $u$.
With this definition, the authors propose a remarkable mechanism,
which they called the {\em probabilistic serial} (PS) algorithm, and prove that it is truthful in the above sense.
At a high level, instead of allocating ``best" items integrally to agents (as RSD does), the PS algorithm does so fractionally obtaining a fractional perfect matching. By the Birkhoff-von Neumann decomposition theorem in combinatorial optimization, this implies a distribution on integral matchings, and this is the output of the PS algorithm. We give a complete description of the PS algorithm in Section \ref{sec:prelim}.

PS and RSD are incomparable and  results on RSD do not a priori imply those for PS. 
There exist instances  where RSD does strictly better than PS, and there exist instances where vice-versa is true (see Appendix \ref{sec:psvsrsd}).
Nevertheless, PS has an ordinal welfare factor of at least a $1/2$ as well.
\begin{theorem}\label{thm:ps-oe}
Given any instance and matching $M^*$, the expected fraction of agents which get an item in PS which is as good as their allocation in $M^*$, is at least a $1/2$.
\end{theorem}

\paragraph{Ordinal Welfare Factor and Popular Matchings}\label{sec:prelim-pop}
Our notion of ordinal welfare factor is somewhat related to the notion of {\em popular} matchings \cite{Gardenfors,AIKM07,KMN10}.
Given preference lists of agents, a matching $M$ is said to be more popular than $M'$
if the number of agents getting {\em strictly} better items in $M$ is at least the number of agents getting strictly better items in $M'$. A matching is popular if no other matching is more popular than it. Thus the notion of popularity partitions the set of agents into three groups while comparing two matching $M$ and $M'$ -- (a) agents that prefer $M$, (b) agents that are neutral, and (c) agents that prefer $M'$ over $M$, where as in the case of ordinal welfare factor, agents that prefer $M$, and agents that are neutral, are not distinguished.

It can be easily seen that any popular matching has an ordinal welfare factor of at least $1/2$, however, (a) not every input instance has a popular matching, and (b) no truthful algorithms are known to compute them when they exist. A few modified measures for matchings  have been studied in the literature~\cite{McC08,HKMN08, KMN10}, namely (a) unpopularity factor (\#agents that are unhappy/\#agents that are strictly happier), and (b) unpopularity margin (\#agents that are unhappy -- \#agents that are strictly happier), (c) popular mixed matchings. The ordinal welfare factor can be seen as a new way of quantifying the popularity of a matching.

\paragraph{Linear Utilities.} We also analyze the performance of RSD and PS mechanisms when agents' utilities are linear, one of the most commonly studied special case of cardinal utilities. 
In this model, we assume that the utility for an agent for his $i^{th}$ preference is $\frac{n-i+1}{n}$.

A quick observation shows that {\em any} serial dictatorship mechanism achieves a welfare of at least $50\%$ of the optimum: the agent at step $t$ gets his $t^{th}$ choice or better giving him a utility of at least $(1 - (t-1)/n)$; thus the total utility is at least $(n+1)/2$, and the optimum is at most $n$. How much better does RSD do? Intuitively, from the above calculation, one would expect the worst case instance would be one where each agent gets one of his top $o(n)$
choices; that would make the optimum value $(1-o(1))n$. We call such instances as {\em efficient} instances since there is an optimum matching where every one gets their (almost) best choice. We show that
for efficient instances, RSD achieves (asymptotically) at least $2/3$ of the social optimum. Furthermore,
there exists instances where RSD does no better. These bounds hold  for PS mechanism as well.
\begin{theorem}\label{thm:rsd-lin}
When the utilities are linear and the instance is efficient, the welfare factor of RSD is at least $2/3 - o(1)$.
Furthermore, there exist efficient instances for which the welfare factor of RSD is at most $2/3 + o(1)$.
\end{theorem}
\begin{theorem}\label{thm:ps-lin}
When the utilities are linear and the instance is efficient, the welfare factor of PS is at least $2/3 -o(1)$.
Furthermore, there exist efficient instances for which the welfare factor of PS is at most $2/3 + o(1)$.
\end{theorem}\noindent
For general instances,  we can show that PS achieves a social welfare of at least $66\%$ (but not quite $2/3$) of the optimum (see Appendix~\ref{app:ps} for details). However, we cannot extend our intuition and prove a similar performance for RSD on general instances; a `balancing' argument shows that the performance of RSD for linear utilities is at least $\perf$, although we conjecture this can be improved (see Appendix~\ref{app:rsd} for details).

\paragraph{Extensions.} We consider two extensions to our model  and focus on the performance of RSD leaving that of PS as an open direction.
In the first, we let the preference lists be incomplete.
The proof of  Theorem \ref{thm:rsd-oe} implies that the ordinal welfare factor of RSD remains unchanged. For linear utilities, we generalize the definition as follows: for an agent with a preference list of length $\ell$, the first choice gives him a utility of $1$ while the last item gives a utility of $1/\ell$. Our results show that RSD doesn't perform very well in this case.

\begin{theorem}\label{thm:rsd-ext-pref}
For linear utilities, RSD gets at least $\tilde{\Omega}(n^{-1/3})$ fraction of the social optimum. Furthermore, there are instances,
where the welfare of RSD is at most $\tilde{O}(n^{-1/3})$ fraction of the social optimum.
\end{theorem}

\noindent
In the second extension, we let the demand of an agent be for sets of size $K$ or less, for some $K\ge 1$. Agents now arrive and pick their best `bundle' among the unallocated items. The ordinal welfare factor of a mechanism is now $\alpha$ if at least an $\alpha$ fraction of agents get a  bundle that is as good (assuming there is a complete order on the set of bundles) as what they got in an arbitrary benchmark allocation. We show that RSD has ordinal welfare factor $\Theta(1/K)$.

\begin{theorem}\label{thm:rsd-ext-muldem}
In the case when each agent has a maximum demand of $K$ items, the ordinal welfare factor of RSD is $\Theta(1/K)$.
\end{theorem}

\subsection{Preliminaries}\label{sec:prelim}
\noindent
{\em Utility Models, Truthful Mechanisms, Welfare Factors.} 
As stated above, we consider two models for utilities of agents.
In the {\em cardinal utility model}, each agent $a$ has a utility function $u_a : I\to {\mathbb R}_{\ge 0}$, with the property that $j \prefers_a \ell$ iff $u_a(j) > u_a(\ell)$. Given a distribution on the matchings, the utility of agent $a$ is $u_a(M) := \sum_{M\in \M} p(M) u_a(M(a))$, where $p(M)$ is the probability of matching $M$. In this paper, we focus on the special case of {\em linear utility model} where
the $i^{th}$ ranked item for any agent $a$ gives him a utility of $(1 - (i-1)/n)$.
We call an instance {\em efficient}, if there is a matching which matched every agent to an item in his top $o(n)$ (for concreteness, let's say this is $n^{1/5}$) choices and thus gives him utility of $(1 - o(1))$.
In the {\em ordinal utility model}, each agent $a$ represents his utility only via his complete preference list $\ge_a$ over the items. 

A mechanism $\A$ is {\em truthful} if no agent can misrepresent his preference and obtain a better item.
In the cardinal utility model this implies that for all agents $a$ and utility functions $u_a,u'_a$, we have
$u_a(M) \ge u_a(M')$ where $M =\A(u_1,\ldots,u_n)$ and $M' = \A(u_1,\ldots,u'_a,\ldots,u_n)$.
%More precisely, let $u_1,\ldots,u_n$ be the true utility functions of the agents, and let $M=\A(u_1,\ldots,u_n)$. %(strictly speaking the input to $\A$ is the preference lists induced by the utility functions.) 
%For any agent $a$, we have that
%$u_a(M) \ge u_a(M')$ for any $M'$ obtained as $\A(u_{-a}, u')$ where agent $a$ changes his utility function to $u'$. 
%In the ordinal utility model, 
In the ordinal utility model, following \cite{bogo-moulin}, we call  a mechanism $\A$ to be truthful if
for all agents $a$ the following holds. Let $M = \A(\prefers_1,\ldots,\prefers_n)$, and let $M' =  \A(\prefers_1,\ldots,\prefers'_a,\ldots,\prefers_n)$ be the solution when $a$ changes his preference list.
Then the randomized allocated to $a$ returned by $M'$ {\em does not} stochastically dominate that returned by $M$. In fact, \cite{bogo-moulin} define a stronger notion of truthfulness in which the allocation to $a$ returned by $M$ stochastically dominates that returned by $M'$. This is a strong notion of truthfulness since it implies truthfulness for {\em any} cardinal utilities as well -- if $M$ stochastically dominates $M'$, then for any utility function $u_a(M) \ge u_a(M')$. 
%
%no agent can misrepresent his preference list, and obtain a matching which (first order) stochastically dominates the matching he gets when he bids truthfully. More precisely, let $M$ be the (distribution on) matchings obtained as $\A(\prefers_1,\ldots,\prefers_n)$, and
%let $M' = \A(\prefers_1,\ldots,\prefers'_a,\ldots,\prefers_n)$. Then $\A$ is weakly truthful iff $M'$ can never stochastically dominate $M$.
%Actually, \cite{bogo-moulin} define a stronger notion of truthfulness which states that if an agent $a$ lies to get matching $M'$, then $M$ stochastically dominates $M'$. One should observe that the stronger notion of truthfulness implies truthfulness in the usual (cardinal) sense: this is because if $M$ stochastically dominates $M'$, then for any utility function $u_a(M) \ge u_a(M')$.

A mechanism has {\em linear welfare factor} of $\alpha$ if for all instances the (expected) sum of linear utilities of agents obtained from the allocation of the mechanism is at least $\alpha$ times the optimal utility allocation for that instance. A mechanism has {\em ordinal welfare factor} of $\alpha$ if for all instances, and for {\em all benchmark matchings} $M^*$, at least $\alpha$ fraction of agents (in expectation) obtain an item via the mechanism which is at least as good as that allocated in $M^*$. \smallskip

\noindent
{\em The Probabilistic Serial Mechanism.}
The {\em probabilistic serial} (PS) mechanism was suggested by Bogomolnaia and Moulin \cite{bogo-moulin}. The mechanism fractionally allocates items to agents over multiple phases, we denote the fraction of the item $i$ allocated to an agent $a$ by $x(a, i)$. These fractions are such that $\sum_{a\in A} x(a,i) = \sum_{i\in I} x(a,i) = 1$ for all agents $a$ and items $i$. Thus this fractional allocation defines a distribution on integral matchings, and this is the distribution returned by PS.

Initially, $x(a, i)=0$ for every agent $a$ and item $i$. We say that an item $i$ is allocated if $\sum_{a\in A} x(a,i) =1$, otherwise we call  it to be available. The algorithm grows $x(a,i)$'s in phases, and in each phase one or more items get completely allocated. During a phase of the algorithm, each agent $a$ grows $x(a, i)$ at the rate of $1$ where $i$ is his best choice in the set of available items. The current phase completes and the new phase starts when at least one item that was available in the current phase, gets completely allocated. The algorithm continues until all items are allocated.

We make a few observations about the above algorithm which will be useful in our analysis.
Firstly note that the algorithm terminates at time $t=1$, at which time
all agents are fractionally allocated one item, that is, $\sum_{i\in I} x(a,i) = 1$. This is because the LHS grows at the rate of $1$ for all agents in any point of time (one uses the fact the preference lists are complete). Secondly, note that any phase lasts for time at least $1/n$ and at most $1$.
%, this is because $0 \le d_i \le n$. 
Therefore, by time $< j/n$ for any $1\le j \le n$, at most
$(j-1)$ phases would've completed.

\subsection{Related Work}\label{sec:rel}

As stated in the introduction, there is a huge amount of literature on matching markets starting with the seminal paper
of Gale and Shapley \cite{GS}, and we refer the reader to numerous detailed surveys (the classic Roth and Sotomayor \cite{RothSot}, and more recent ones by S\"onmez and \"Unver \cite{SU-survey} and Abdulkadiroglu and S\"onmez \cite{AS10})
for a more elaborate picture. In this section, we try to review the papers which are most relevant to our paper.

The one-sided matching market design problem was first studied by
Hylland and Zeckhauser \cite{HZ} who propose a mechanism to find a distribution on matchings
via a market mechanism where each agent is given equal, artificial budgets. Their mechanism
returns Pareto optimal, envy-free solutions; unfortunately it is not truthful.
Zhou \cite{Zhou}, answering a question of Gale \cite{Gale}, showed that there can be no truthful mechanism which is anonymous/neutral and satisfies {\em ex ante} Pareto optimality; this is a stronger notion of ex post Pareto optimality which we have discussed before. 
%A (distribution on) matching(s) is ex ante Pareto optimal, if there is no other (distribution on) matching(s) which gives  strictly higher expected utility for some agent, and no less expected utility for other agents. 
%RSD is not ex ante Pareto optimal; it is anonymous and truthful. 
Svensson \cite{Svensson} showed that serial dictatorship mechanisms are the only truthful mechanisms which are (ex post) Pareto optimal, non bossy, and anonymous.
The study of mechanisms with ordinal utilities for this problem was started by Bogomolnaia and Moulin\cite{bogo-moulin}.
The probabilistic serial mechanism was in fact proposed in an earlier paper on scheduling jobs by Cres and Moulin \cite{cres-moulin}. Following the work of \cite{bogo-moulin}, there was a list of work characterizing stochastic dominance \cite{AS03, bogo-moulin2}, and generalizing it to the case of incomplete preference lists \cite{Katta-Seth}, and  to multiple copies of items \cite{BudYM}.

A generalization of the one-sided matching market design problem where the items are initially owned by agents, and they
need to be exchanged among themselves is called the {\em market exchange} problem (or the {\em housing market} problem.)
This problem has many applications, the most important of which is the kidney exchange problem \cite{RothSU}.
The exchange problem was introduced by Scarf and Shapley \cite{ScarfShapley} who were interested in {\em core} allocations; allocations where no subset of agents can re-allocate their items and everyone be better off. They showed the core allocation is non-empty via a constructive algorithm attributed to David Gale; this is Gale's {\em top trading cycle} (TTC) algorithm which
is truthful and Pareto optimal. There has been a lot of work extending TTC to various scenarios and characterizing truthful mechanisms in this model (we refer the interested reader to the surveys \cite{SU-survey} and \cite{AS10} and the references therein); however the result most relevant to work is the one due to Abdulkadiroglu and S\"onmez \cite{AS98}. They showed that in the one-sided market mechanism problem, if one assumes a random initial endowment and runs the TTC algorithm, then the behaviour (distribution on matchings) is exactly the same as that of RSD.

\iffalse
As stated in the introduction, our notion of ordinal welfare factor is closely related to popular matchings. Popular matchings
were first defined for deterministic matchings by Gardenfors~\cite{Gardenfors}, however, there are instances without any popular matchings. Abraham et al \cite{AIKM07} give a polynomial time algorithm to decide if an instance has a popular matching or not. More recently, Kavitha et al\cite{KMN10} studied popular mixed matchings, and they showed that they always exist and gave an algorithm to find them. At this point, we remark that for the complete preference list case, our result that PS has ordinal welfare factor $1/2$ implies that PS returns a popular mixed matching. There are other notions of matchings which have been studied by computer scientists, for instance \cite{IKMMP06,ACKM}.
\fi
\iffalse
As stated in the introduction, our notion of ordinal welfare factor is closely related to popular matchings. 
Popular matchings were first defined for deterministic matchings by Gardenfors~\cite{Gardenfors}. Abraham et al \cite{AIKM07} give a polynomial time algorithm to decide if an instance has a popular matching or not. Kavitha et al\cite{KMN10} who showed that there always exist popular mixed matchings, and gave an algorithm to find them. Matchings with bounded unpopularity has been studied by McCutchen~\cite{McC08} and Huang~{\em et al}~\cite{HKMN08}.
\fi

Finally, the study of mechanism design without money has been of recent interest in the computer science community \cite{ProcTenn, Dekel}. We already have mentioned the relation to popular matchings in the introduction.
There has been works motivated by the market exchange problem \cite{MandM, SD}, and item allocation problem \cite{Con1,Con2}, however none of them address the problem that we study.

\section{Ordinal Welfare Factor of RSD and PS Mechanisms}\label{sec:rsd}
In this section, we prove Theorems \ref{thm:rsd-oe} and \ref{thm:ps-oe}.
%We first consider the ordinal welfare factor. Before doing so,
We first show that the ordinal welfare factor of {\em any} mechanism
is at most $1/2$ in the instance where every agent has the same preference list.
%This shows a bottleneck in the measure.

\def\D{\mathcal D}
\begin{theorem}\label{thm:rsd-oe-lb}
If every agent has the same preference list $(1,2,\ldots,n)$, then the ordinal welfare factor of any mechanism
is at most $1/2 +1/2n$.
\end{theorem}
\begin{proof}
A mechanism returns a probability distribution on matchings
which we will interpret as a distribution of permutations, let $\D$ be that distribution.
We choose the benchmark matching $M^*$ to be a random matching, i.e. a random permutation $\pi$ of the agents.

%Consider a reference matching $M^*$ that assigns items to agents according to a permutation $\pi$ on agents, chosen uniformly at random.
%We now fix the ``optimum" matching $M^*$. Recall that in ordinal welfare factor we are competing against any matching. In fact, the matching $M^*$ will be a random matching, that is a random permutation $\pi$ of the agents.
It suffices to show that for any fixed permutation $\sigma\in \D$, the expected number of agents $a$ such that
$\sigma(a) > \pi(a)$ is $(n-1)/2$. The bound then holds for $\D$. Since $\pi$ is chosen uniformly at random, the probability
that $\pi(a) \le \sigma(a)$ is precisely $\sigma(a)/n$.
%(We are abusing notation by using $\sigma(a),\pi(a)$ for both the index of the item $a$ gets as well as an integer).
So the expected number of happy people for the permutation $\sigma$ is $\sum_{a\in A}\sigma(a)/n = (n+1)/2$.
\end{proof}

\subsection{Ordinal Welfare Factor of RSD}
In this section, we prove Theorem \ref{thm:rsd-oe}. 
%In the rest of  the section, we prove Theorems \ref{thm:rsd-oe} and \ref{thm:ps-oe}.
%\begin{proofof}{Theorem \ref{thm:rsd-oe}
Let $M^*$ be the unknown benchmark matching.
%matching, and for agent $a$, let $M(a)$ be the item he gets in $M$. We consider the run of RSD with agent $i$ arriving at time $i$.
We call an agent $a$ {\em dead} at time $t$ if he hasn't arrived yet and all items as good as $M^*(a)$ in his preference list has been allocated.
%
%$M^*(a)$ has already been allocated (note that I am using a stronger notion of death
%than used in the paper.).
Let $D_t$ be the expected number of dead agents at time $t$. Let $\alg_t$ be the expected number of agents
who get an item as good as their choice in $M^*$ by time $t$. From the above definition, we get
\begin{equation}\label{eq:ALGandD}
\alg_{t+1} - \alg_t = 1 - \frac{D_t}{n-t}
\end{equation}
%
% If $h_t$ is the expected
%number of happy people from $1$ to $t$, it is clear that $\frac{dh_t}{dt} \ge 1 - \frac{d_t}{n-t}$.
%We will now show the following lemma which will imply that RSD has efficiency $\frac{1}{2} - o(1)$.
We will now bound $D_t$ from above which along with \eqref{eq:ALGandD} will prove the theorem.

\begin{lemma}\label{lem:boundingD}
$D_{t} \le \frac{(t+2)(n-t)}{n+1}$ for $1\le t\le n$.
\end{lemma}

\noindent
Before proving the lemma, note that adding \eqref{eq:ALGandD} for $t=1$ to $n-1$ gives
$\alg_n - \alg_1 \ge \sum_{t=1}^{n-1} \left(1 - \frac{t+2}{n+1}\right)$, implying $\alg_n-\alg_1 \ge n/2 - 2n/n$. This proves that the ordinal welfare factor of RSD is at least $1/2 - o(1)$ proving Theorem \ref{thm:rsd-oe}.
\def\all{{\tt ALL}}
\def\late{{\tt LATE}}
\def\want{{\tt WANT}}
\begin{proof}
Let us start with a few definitions. For an item $i$ and time $t$, let $\all_{i,t}$ be the event that item $i$
is allocated by time $t$. For an agent $a$ and time $t$, let $\late_{a,t}$ be the event that $a$ arrives
after time $t$. The first observation is this: if an agent $a$ is dead at time $t$, then the event $\all_{M(a),t}$ and $\late_{a,t}$ must have occurred. Therefore we get
\begin{equation}\label{eq:boundingD}
D_t \le \sum_{a\in A} \Pr[\all_{M(a),t} ~\wedge~ \late_{a,t}]
\end{equation}

\iffalse
%It is easy to bound $\Pr[\late_{a,t}]$; it is precisely $1 - t/n$.
\noindent
We now calculate $\Pr[\all_{i,t}]$. Given an item $i$, let $W_{i,t}$ be the random variable indicating the number of agents yet to come who have $i$ as their best choice in the set of available items at time $t$.
Let $\want_{i,t} = \Exp[W_{i,t}]$.
Since the preference lists are complete, every remaining agent has a best choice, and so we get that
\begin{equation}\label{eq:want}
\sum_{i\in I} \want_{i,t} = (n - t)
\end{equation}
\noindent
\begin{claim}
$\Pr[\all_{i,t}] = \sum_{t'\le t}\frac{\want_{i,t'}}{n-t'}$.
\end{claim}
\begin{proof}
In fact, we claim that an item $i$ is allocated {\em precisely at} time $t$ is $\want_{i,t}/(n-t)$. This is because, the probability that the item $i$ is allocated at time $t$ {\em conditioned} on the event that $i$ has not been allocated yet (that is, $\bar{\all_{i,t-1}}$), is precisely $\Exp[W_{i,t} | \bar{\all_{i,t-1}}]/(n-t)$. But if $i$ is allocated before time $t$, we have $W_{i,t} = 0$. So, $\Exp[W_{i,t}|\bar{\all_{i,t-1}}] = \want_{i,t}/(1 - \Pr[\all_{i,t-1}])$.
So, the probability that $i$ is allocated at time $t$ {\em and} not before is precisely $\want_{i,t}/(n-t)$.
\end{proof}
\fi
\noindent
Note that $\Pr[\late_{a,t}]$ is precisely $(1 - t/n)$. Also, note that $\sum_{i\in I} \Pr[\all_{i,t}] = t$.
This is because all agents are allocated {\em some} item.
Now suppose {\em incorrectly} that $\all_{M(a),t}$ and $\late_{a,t}$ were independent. Then, \eqref{eq:boundingD} would give us

\begin{align}
D_t \le (1 - \frac{t}{n}) \sum_{a\in A}  \Pr[\all_{M(a),t}] = (1 - \frac{t}{n}) \sum_{i\in I}  \Pr[\all_{i,t}]  = \frac{t(n-t)}{n}
\end{align}
which is at most the RHS in the lemma. However, the events are not independent, and one can construct examples where the above bound is indeed incorrect. To get the correct bound, we need the following claim.

\begin{claim}\label{claim:dep}
$$\frac{\Pr[\all_{M(a),t} ~ \wedge ~ \late_{a,t}]}{(n-t)} \le \frac{\Pr[\all_{M(a),t+1} ~\wedge ~\bar{\late_{a,t+1}}] }{(t+1)}$$
\end{claim}
\begin{proof}
This follows from a simple charging argument. Fix a relative order of all agents other than $a$
and consider the $n$ orders obtained by placing $a$ in the $n$ possible positions.
Observe that if the event $\all_{M(a),t} \wedge \late_{a,t}$ occurs at all, it occurs exactly $(n-t)$ times when $a$'s position is $t+1$ to $n$. Furthermore, crucially observe that if the position of $a$ is $1$ to $t+1$, the item $M(a)$ will still be allocated. This is because the addition of $a$ only leads to worse choices for agents following him and so if $M(a)$ was allocated before, it is allocated even now. This proves that for every $(n-t)$ occurrences of $\all_{M(a),t} \wedge \late_{a,t}$, we have $(t+1)$ occurrences of
the event $\all_{M(a),t+1} \wedge \bar{\late_{a,t+1}}$.
The claim follows as it holds for every fixed relative order of other agents.
%and the events $\all_{M(a),t} ~ \wedge ~ \late_{a,t}$ and
\end{proof}
\noindent
Now we can finish the proof of the lemma. From Claim \ref{claim:dep}, we get
$$\frac{t+1}{n-t}\cdot\Pr[\all_{M(a),t} ~ \wedge ~ \late_{a,t}] \le \Pr[\all_{M(a),t+1}] - \Pr[\all_{M(a),t+1} ~\wedge ~\late_{a,t+1}]$$
Taking the second term of the RHS to the LHS, adding over all agents, and invoking \eqref{eq:boundingD}, we get
\begin{equation}\label{eq:drec}
\frac{t+1}{n-t}\cdot D_t + D_{t+1} \le  t+1
\end{equation}
Using the fact that $D_{t+1} \ge D_t - 1$ (the number of dead guys cannot decrease by more than $1$), and rearranging, proves the lemma.
\end{proof}

\subsection{RSD and online bipartite matching}\label{sec:prelim-kvv}
In this section, we highlight the relation between RSD and algorithms for online bipartite matching.
As we show below, the analysis of RSD above can be seen as a generalization of online bipartite
matching algorithms. This section can be skipped without loss of continuity.

In the online bipartite matching problem, vertices of one partition (think of them as agents) are fixed while
vertices of the other partition (think of them as items) arrive in an adversarial order. When an item arrives, we get to see the incident edges on agents. These edges indicate the set of agents that desire this item.
%Think of these edges as indicating agents who desire this item.
The algorithm must immediately match this item  to one of the unmatched agents desiring it  (or choose to throw it away). In the end, the size of the obtained matching is compared with the optimum matching in the realized graph.

Karp, Vazirani and Vazirani \cite{KVV} gave the following algorithm (KVV) for the problem: fix a {\em random} ordering of the agents, and when an item arrives give it to the first unmatched agent in this order.
They proved\footnote{In 2008, a bug was found in the original extended abstract of \cite{KVV}, but was soon resolved. See \cite{GM08,BM08,GGKM11} for discussions and resolutions.} that the expected size of the matching obtained is at least $(1-1/e)$ times the optimum matching. For simplicity (and indeed this can be shown to be without loss of generality), assume that the size of the optimum matching is indeed $n$, the number of agents and items.
The KVV theorem can be `flipped around' to say the following. Suppose each agent has the %(same) 
preference list which goes down its desired items in the order of entry of items. Then, if agents arrive in a random order and pick their best, unallocated, desired item, in expectation an $(1-1/e)$ fraction of agents are matched.
That is, if we run RSD on this instance (with incomplete lists), an $(1-1/e)$ fraction of agents will get an item.

\noindent
We should point out that the above result does not a priori imply an analysis of RSD, the reason being that in our problem an agent $a$, when he arrives, is allocated an item even if that item is {\em worse}
than what he gets in the benchmark matching $M^*$. This might be bad since the allocated item 
could be `good' item for agents to come. In particular, if the order chosen is not random but arbitrary, the performance of the algorithm is quite bad; in contrast, the online matching algorithm still has a competitive ratio of $1/2$.
Concretely, consider the following instance: Agent $i$ for $2 \le i \le n$ have item $(i-1)$ as their first choice, and item $i$ as their second choice. Agent $1$ has item $1$ as his first choice and item $n$ as his second choice.  $M^*$ matches agent $2,\ldots,n$ to items $1,\ldots,(n-1)$ and agent $1$ to item $n$. Now if agents arrive in order $1$ to $n$, only agent $1$ does as well as its allocation in $M^*$.
We now tailor our analysis of RSD to give another proof of the online bipartite matching result.

\begin{corollary}\label{cor:kvv}
KVV matches at least $(1-1/e)n - 1$ agents in expectation.
\end{corollary}
\begin{proof}
The analysis is exactly same as the proof of Theorem~\ref{thm:rsd-oe} except we have the following 
`savings' -- instead of $\sum_{i\in I} \Pr[\all_{i,t}] = t$, we have $\sum_{i\in I} \Pr[\all_{i,t}] = \alg_t$.
This is because by time $t$ the number of items allocated is precisely the number of agents who receive an item, which is $\alg_t$. \eqref{eq:drec} then evaluates to 
$\frac{t+1}{n-t}\cdot D_t + D_{t+1} \le \alg_{t+1}$. 
Using the fact that $D_{t+1} \ge D_t - 1$ and $\alg_{t+1} \le \alg_t + 1$, we get
$\frac{1}{n-t}\cdot D_t  \le \frac{1}{n+1}\cdot \alg_t + \frac{2}{n+1}$, which along with \eqref{eq:ALGandD} gives us
$$\alg_{t+1} \ge \left(1 - \frac{1}{n+1}\right)\alg_t + \left(1 - \frac{2}{n+1}\right)$$
This solves to $\alg_n \ge (1 - 1/e)\cdot (n - 1)$.
\end{proof}

\subsection{Ordinal Welfare Factor of PS}
In this section, we prove Theorem \ref{thm:ps-oe}. We suggest the reader to refer to the algorithm
and its properties as described in Section \ref{sec:prelim}. In particular, we will use the following
observation. \smallskip

\noindent
{\bf Observation 1:}
By time $< j/n$, for any $1\le j\le n$, at most $(j-1)$ items are completely allocated.
\smallskip

\noindent
Let $M^*$ be the unknown benchmark matching.
For an agent $a$, let $t_a$ be the time in the run of the PS algorithm
at which the item $M^*(a)$ is completely allocated.
Observe that the probability agent $a$ gets an item $M^*(a)$ or better
is precisely $t_a$, since till this time $x(a,i)$ increases for items $i\ge_a M^*(a)$.
Summing up all agents, we see that the ordinal welfare factor of the
PS mechanism is $\sum_a t_a$.
The observation above implies at most $(j-1)$ agents have $t_a < j/n$. So,
$\sum_a t_a \ge \sum_{j=1}^n (n-j+1)/n \ge n/2 + 1/2$. This completes the proof of Theorem \ref{thm:ps-oe}.
\noindent

\section{Linear Welfare Factor of RSD and PS}\label{sec:ps}
In this section, we establish bounds on the linear welfare factor of RSD and PS mechanisms.
We first prove Theorem \ref{thm:rsd-lin} in two lemmas. Recall that an instance is called efficient if there exists a matching in which every agent is matched to an item in his top $o(n)$ choices.

\begin{lemma}
When the utilities are linear and the instance is efficient, the linear welfare factor of RSD is at least $\left(2/3 - o(1)\right)$.
\end{lemma}
\begin{proof}
The proof almost directly follows from Lemma \ref{lem:boundingD}. Let $U_t$ denote the expected utility obtained by time $t$. Consider the agent coming at time $t+1$. If he is not dead already, then he will get a utility of at
least $(1 - o(1))$ (since the instance is efficient). If he is dead, then he will get a utility of at least $(1 - t/n)$.
This is because only $t$ items have been allocated and this agent takes an item $(t+1)$th ranked or higher.
Therefore,
$$ U_{t+1} - U_t \ge \left(1 - \frac{D_t}{n-t}\right)\cdot (1 - o(1)) + \frac{D_t}{n-t}\cdot\left(1 - t/n\right) \ge 1 - o(1)  - \frac{t}{n}\cdot\frac{D_t}{n-t}$$
\noindent
Using Lemma \ref{lem:boundingD}, we get $U_{t+1} - U_t \ge 1 - o(1) - \frac{t(t+2)}{n(n+1)}$. Summing over all $t$, we get that the total utility of RSD is at least $(1 - o(1))n - (n/3 + o(n)) = (2/3 - o(1))n$.
\end{proof}
The above analysis can be modified via a balancing trick to give a strictly better than $50\%$ guarantee for all instances in the case of linear utilities. However, the improvement we get is small, and the analysis is not that informative. We defer it to the Appendix \ref{app:rsd}.

\begin{lemma}\label{thm:rsd-lin-lb}
When the utilities are linear, there exists an efficient instance for which RSD gets a utility of at most $(2/3 + o(1))n$.
\end{lemma}
\begin{proof}
Partition $n$ agents and items into $t$ blocks 
of size $n/t$ each,
where $t = n^{1/5}$.
We denote the $j^{th}$ block of agents and items by $A_j$ and $I_j$ respectively, and they number from $\left(\frac{(j-1)n}{t}+1\right)$ to $\frac{jn}{t}$.

We now illustrate the preference lists of agents. Fix an agent $a$ in block $A_j$. Let he be the $k^{th}$ agent in the block,
where $1\le k\le n/t$, i.e. his agent number is $(j-1)n/t + k$.
A random set of $t^3$ items is picked from each of blocks $I_1,\ldots,I_{j-1}$, and these form the first $(j-1)t^3$
items in his preference list, in increasing order of item number. The item $(j-1)n/t + k$ is his $((j-1)t^3 + 1)^{th}$ choice.
His remaining choices are the remaining items considered in increasing order. This completes the description of the
preference lists of the agents.

Note that if every agent $a$ is assigned the corresponding item with the same number, then each agent gets one of his top
$t^4$ choices, leading to a utility of at least $(1 - \frac{t^4}{n}) = 1 - o(1)$. So, the instance is indeed efficient. We now show that RSD gets utility at most $2n/3 + o(1)$.

 Let $\sigma$ be a random permutation of the agents. We divide $\sigma$ into $t$ chunks of $n/t$ agents, with
 the $j^{th}$ chunk, $S_j$, consisting of agents $\sigma(\frac{(j-1)n}{t}+1)$ to $\sigma(\frac{jn}{t})$.
 Note that with high probability ($\ge (1-1/t^3)$), we have that for any block $A_j$ and chunk $S_i$,
$|A_j \cap S_i| \in \left[(1 - \frac{1}{t^2})\frac{n}{t^2}~,  (1 + \frac{1}{t^2})\frac{n}{t^2}\right]$. We now state the crucial observation.
 \begin{claim}
With high probability, at least $(1-\frac{1}{t^3})$ fraction of the items in the first $i$ blocks have been allocated after arrival of first $i$ chunks.
 \end{claim}
\begin{proof}
Fix an agent $a$ in the first $i$ chunks and consider the allocation to the agent $a$ when he arrives.
If more than $in/t^3$ items in set of items in $I_1\cup\cdots\cup I_i$ are left unallocated, then we claim
that the probability that the agent $a$ is not allocated one of these free items is exponentially small. Note that $it^3$
items of the $in/t$ items in the first $i$ blocks form the top choices of agent $a$. Furthermore, this set is chosen {\em independently}  of all other agents. In particular, the probability that none of the free items lie as one of these choices is at most
$\left(1 - \frac{t^4}{n}\right)^{in/t^3} \le e^{-t}$. Since $t$ is large enough (assumed to be $n^{1/5}$), this probability is small enough to apply union bound over all agents.
\end{proof}

Now we are ready to analyze RSD. Consider the $(i+1)^{th}$ chunk of agents. With high probability, there are at least
$\frac{n}{t^2}(1 - \frac{1}{t^2})$ agents from each block $A_1, \ldots, A_i$ in $S_{i+1}$. Since only $in/t^3$ items remain
from the first $i$ block of items, at least $\frac{in}{t^2}(1 - \frac{1}{t^2}) - \frac{in}{t^3}$ of these agents must get an item
from blocks $(i+1)$ or higher. However, this gives them utility at most $(1 - \frac{in/t}{n}) \ge 1 - i/t$.
That is, the {\em drop} in their utility to what they get in the optimum is at least $i/t$.  Summing the total drop over all agents and all chunks, we get that the difference between RSD and the optimum is at least
$$\sum_{i=1}^t \frac{in}{t^2}(1 - \frac{1}{t})\frac{i}{t} = (1- o(1))\frac{n}{t^3}\sum_{i=1}^n i^2 = n/3 $$
Therefore, the social welfare of RSD is at most $(2/3 + o(1))n$.
\end{proof}
\noindent

\medskip
\paragraph{Linear Welfare Factor of PS Mechanism}
We now establish the upper and the lower bound on the linear welfare factor of the PS mechanism.
As in the case of RSD, we focus on efficient instances
in the main body since they contain the main analysis ideas. 

\begin{lemma}
When the utilities are linear and the instance is efficient, the linear welfare factor of the PS mechanism is at least $\left(2/3 - o(1)\right)$.
\end{lemma}
\begin{proof}
Let $o_a$ denote the utility obtained by agent $a$ in the utility optimal matching.
Since the instance is efficient, $o_a = 1 - o(1)$ for all agents $a$.

Consider the $j^{th}$ phase of PS, and suppose it lasts for time $\Delta_j$. Observation 1 implies that $\sum_{j\le \ell} \Delta_j \ge \ell/n$. Furthermore, in phase $j$, at least $(n- j +1)$ agents are getting utility at a rate higher than their utility in the optimal matching.
This is because at most $(j-1)$ items have been completely allocated. Furthermore, the remaining $(j - 1)$ agents
are getting utility at a rate at least $(1 - (j-1)/n)$ since they are growing their $x(a,i)$ on their $j^{th}$ choice or better.
So, the total utility obtained by PS is at least
$\sum_{j=1}^n \Delta_j\cdot\left((n-j+1)\cdot(1 - o(1)) + (j - 1)\cdot(1 - \frac{j-1}{n}) \right) \approx \sum_{j=1}^n \Delta_j\left(\frac{n^2 - (j-1)^2}{n}\right)$

\indent
Now, the RHS is smallest if $\Delta_1$ is as small as possible, modulo which, $\Delta_2$ is as small as possible and so on.
Given the constraint on $\Delta_j$'s, we get that the RHS is at least $\sum_{j=1}^n \frac{n^2 - (j-1)^2}{n^2} \approx 2n/3$.
This implies a $2/3$ approximation.
\end{proof}
%
%We now show a matching upper bound on the performance of PS. In particular, we construct an efficient instance
%in which PS attains social welfare at most $(2/3 + o(1))n$.
%
\begin{lemma}
When the utilities are linear, there exists an efficient instance such that the total utility of PS is at most $(2/3 + o(1))$ times the optimum.
\end{lemma}
\begin{proof}
%of}{Theorem \ref{thm:ps-lin}, part 2}
The construction of the lower bound example is similar to (but not the same as) that given in the proof of Lemma
\ref{thm:rsd-lin-lb}. We have $n$ agents and items divided in $t$ blocks (we assume $t$ divides $n$) of $n/t$ agents each, $t$ will be $o(n)$ in the construction.
We denote the $j^{th}$ block of agents and items by $A_j$ and $I_j$ respectively, and they number from $\left(\frac{(j-1)n}{t}+1\right)$ to $\frac{jn}{t}$.

The agents' preference lists are as follows. Fix an agent $a$ in block $A_j$. Let he be the $k^{th}$ agent in this block,
where $1\le k\le n/t$, i.e. his agent number is $(j-1)n/t + k$. The first $j$ preferences of this agent are precisely
the $k^{th}$ items in the blocks $I_1$ to $I_j$, in the increasing order of the item number. That is, for $1\le \ell \le j-1$, the $(\ell  + 1)$th choice of agent $a$ is
the item number $\ell n/t + k$. Furthermore, for $j < l \le t$, the item number $\ell n/t + k$, present in block $I_\ell$, is the
$(\ell n/t + 1)^{th}$ choice of agent $a$. Thus, so far we have described where $t$ items reside in the preference list of $a$.
The remaining $n - t$ choices of his are all the remaining items from $1$ to $n$, in that order.

\begin{lemma}\label{lem:1}
The optimum social welfare of the instance is at least $n - t$.
\end{lemma}
\vspace{-3mm}
\begin{proof}
Note that agent $a$ whose number is $(j-1)n/t + k$ can be assigned item $(j-1)n/t + k$, which is his $j^{th}$ choice.
So, in this matching, every agent gets his $t^{th}$ choice or higher, leading to a social welfare of at least
$n(1 - t/n)$.
\end{proof}

\begin{lemma}\label{lem:2}
In this instance, PS attains a social welfare of at most $(2/3 + o(1))n$.
\end{lemma}
\begin{proof}
The following claim is the key observation in the proof of this lemma.
\begin{claim}
In the run of PS, each phase is of duration $1/t$, and all items in set $I_j$ are completely allocated
in phase $j$.
\end{claim}
\noindent
Note that initially all agents prefer items in $I_1$, furthermore, each item in $I_1$ is the best choice of exactly $t$
agents. So at time $1/t$, all items in $I_1$ are completely allocated.
Subsequently, the best choice of every agent among the remaining items lies in $I_2$ and arguing inductively proves the claim. \end{proof}
Now we can complete the proof of the lemma. Note that each agent in block $A_j$ obtains precisely $1/t$ fraction of an item from
every block. For the first $j$ blocks, he gets utility at most $1$ per item. Subsequently, for $j < \ell \le t$, he gets utility
$(1 - \frac{(\ell - 1)(n/t)}{n})$. So, the utility of an agent in $A_j$ is at most
$\frac{1}{t}\left( 1\cdot j + \sum_{\ell = j+1}^t \frac{t - \ell + 1}{t}\right) \approx \frac{1}{t}\left( j + (t - j) - \frac{t^2 - j^2}{2t}\right)
 = \frac{1}{2} + \frac{j^2}{2t^2}$
 where the $\approx$ hides the additive $o(1)$ term.
Summing over all blocks $j$ and all agents in a block we get that the social welfare of PS is at most
$\frac{n}{2} + \frac{n}{6} + o(n) = 2n/3 + o(n)$.
\end{proof}
\noindent
%Lemmas \ref{lem:1} and \ref{lem:2} imply Theorem \ref{thm:ps-lin-ub}.
%\end{proof}
\noindent
For PS, we can do a much better analysis for general instances with linear utilities, than what we can do for RSD. We defer this to the Appendix \ref{app:ps}.

\section{Extensions}
\label{sec:extn}
In this section, we study two extensions to the model studied thus far, (a) incomplete preference list: where for every agent, there exists a set of items $S$ such that he has a complete preference list over items from $I-S$ and he prefers to stay unmatched over an assignment of an item from $S$, and (b) relaxation of the unit-demand assumption: where an agent has demand only for subsets of size $K$ and he has a complete preference list over all subsets of size $K$, for some parameter $K$. In both cases, we analyze the performance of RSD whose definition can be easily modified to both cases. We leave the extension and analysis of PS to these cases as an open direction.

\subsection{Incomplete Preference Lists}

We assume that there are $n$ agents and $m$ items where each agent $a$ has an
{\em ordered} preference list $L_a$
on some arbitrary subset of $\{ 1,2, ..., m \}$.
We note that the proof of Theorem~\ref{thm:rsd-oe} can be modified easily to show that the ordinal welfare factor of RSD is at least $1/2$ even in this case.

Now we focus on the linear utility model. In this model, if agent $a$ is allocated the $j^{\rm th}$ item
on the list, denoted by $L_a(j)$, it gets a utility of $\frac{|L_a| +1 - j }{ |L_a|} $.
Our main result is that the linear welfare factor of RSD drops to $\tilde{\Theta}(1/n^{1/3})$.
 We now focus here on establishing almost tight
upper and lower bounds on the linear welfare factor of RSD.

\medskip
\noindent
{\bf  Linear Welfare Factor of RSD with partial lists is $\tilde{O}(1/n^{1/3})$:}

We start by showing a family of instances that establish an upper bound of $\tilde{O}(1/n^{1/3})$
on the linear welfare factor of RSD.
The set $A$ of agents is partitioned into two sets, namely, a
set $G = \{ g_1, g_2, ..., g_a \}$ of {\em good} agents, and
a set $B = \{ b_1, b_2, .., b_b \}$ of {\em bad} agents, where $a = n^{2/3}$
and $b = n - n^{2/3}$. The set of items is partitioned into two sets $I_G = \{ 1, 2, ..., p \}$ and
$I_B = \{ 1', 2', ..., q'\}$ where $p = n^{2/3}$ and $q = n^{1/3}$. For each $i \in [1..a]$, the
preference list of agent $g_i$ consists of a single item, namely, $\{ i \}$.
For each $j \in [1..b]$, the preference list of agent $b_j$ consists of items in $I_B$
ordered as $\{ 1', 2', ..., q'\}$, followed by the item $i \in I_G$ such that
$i = 1 + (j \bmod n^{2/3})$. An agent $b_j \in B$ whose preference list ends in an item
$i \in I_G$, is referred to as a {\em type-$i$} bad agent; there are $(n^{1/3}-1)$ type-$i$
bad agents for each $i \in [1..a]$.

It is easy to verify that
$\opt$ is at least $n^{2/3}$, since any allocation that assigns items in $I_G$ to agents in $G$ in accordance with the unique matching defined by the preference lists of agents in $G$
achieves $n^{2/3}$ utility.
We now show that in contrast, the expected utility of RSD on this instance is bounded by $O(n^{1/3} \log n)$, yielding the desired lower bound.

Let $t_1 = 2 n^{2/3}$ and let $t_2 = 9 n^{2/3} \ln n$.
The analysis relies on the probability of occurrence (over the choice of random
arrival ordering of the agents) of the four events defined below. We bound
the probabilities of these events by careful application of Chernoff's inequality.
% on appropriately
%defined variable; we omit the details of these standard arguments in this version.

\begin{itemize}
\item[${\cal E}_0$:]
We say that the event ${\cal E}_0$ occurs if there exists an $i \in [1..a]$ such that the number of
type-$i$ bad agents that arrive in the first $t_1$ steps exceeds $(n^{1/3}-1)/2$. Then
%we can show that
$$ {\rm Pr}[{\cal E}_0] \leq  (n^{2/3}) e^{-\Omega(n^{1/3})} =  e^{-\Omega(n^{1/3})} $$

\item[${\cal E}_1$:]
We say that the event ${\cal E}_1$ occurs if the number of good agents that arrive
in the first $t_1$ steps exceeds $n^{1/3}$. Then
$ {\rm Pr}[{\cal E}_1] \leq   e^{-\Omega(n^{1/3})}. $

\item[${\cal E}_2$:]
We say that event ${\cal E}_2$ occurs if there exists an $i \in [1..a]$ such that  no type-$i$ bad agent arrives in the interval $(t_1,t_2]$. Then

$ {\rm Pr}[{\cal E}_2] = {\rm Pr}[{\cal E}_2 ~|~{\cal E}_0] \cdot {\rm Pr}[{\cal E}_0]
+ {\rm Pr}[{\cal E}_2 ~|~\overline{{\cal E}_0}] \cdot {\rm Pr}[\overline{{\cal E}_0}]
\leq {\rm Pr}[{\cal E}_0] + {\rm Pr}[{\cal E}_2 ~|~\overline{{\cal E}_0}]  \leq 1/n^2$

\item[${\cal E}_3$:]
We say that the event ${\cal E}_3$ occurs if the total number of good agents  that arrive
in the interval $[1,t_2]$ exceeds $18 n^{1/3} \ln n$. Then
$ {\rm Pr}[{\cal E}_3] \leq  1/n^{2}.$

\end{itemize}
\noindent
Thus $ {\rm Pr}[ {\cal E}_1 \vee {\cal E}_2 \vee {\cal E}_3] \leq  3/n^{2}.$

We now analyze the performance of RSD when none of the events
${\cal E}_1$, ${\cal E}_2$, or ${\cal E}_3$ occur.
If the event ${\cal E}_1$ does not occur, then by time $t_1$, at least $n^{1/3}$ agents from
$B$ must have
arrived, so by time $t_1$, all items in $I_B$ must have been assigned to agents in $B$.
Total contribution from items in $I_B$ is thus bounded by $n^{1/3}$.
Additionally, if the event ${\cal E}_2$ does not occur, we are guaranteed that by time $t_2$,
all items in $I_G$ must have been assigned to agents in $B \cup G$.
Finally, if the event ${\cal E}_3$ also does not occur, then at most $18 n^{1/3} \ln n$
items in $I_G$ get assigned to agents in $G$, and the rest are assigned to agents in $B$.
Thus the total contribution from items in $I_G$ can be bounded by $O(n^{1/3} \log n)$.
Hence the expected utility of RSD is bounded by\\

\noindent
$ {\rm Pr}[ {\cal E}_1 \vee {\cal E}_2 \vee {\cal E}_3] \cdot (n^{1/3}+n^{2/3}) + {\rm Pr}[ \overline{( {\cal E}_1 \vee {\cal E}_2 \vee {\cal E}_3 )} ] \cdot O(n^{1/3} \log n) = O(n^{1/3} \log n).$

\noindent
This completes the proof that the linear welfare factor of RSD is at most $\tilde{O}(1/n^{1/3})$.

\medskip

\medskip
\noindent
{\bf  Linear Welfare Factor of RSD with incomplete preference lists is  $\tilde{\Omega}(1/n^{1/3})$:}

\noindent
We now show that the linear welfare factor of RSD is at least
$\tilde{\Omega}(1/n^{1/3})$. We start by making two easy
observations, namely, (a) $\opt$ is at most $n$, (b) linear welfare factor  of RSD is
$\max\{ \Omega(1/\opt), \Omega(\opt/n) \}$.
It thus suffices to show that
linear welfare factor of RSD is $\tilde{\Omega}(1/n^{1/3})$ when $n^{1/3} \le \opt \le n^{2/3}$.

Let $G$ and $I_G$ respectively denote the subset of agents and items
that contributes to $\opt$, that is, there exists a perfect matching $M$ between $G$ and $I_G$.
Note that $|G| = |I_G| \ge \opt$. We start with the assumption
that each agent in $G$ contributes $\Omega(1)$ -- this case essentially captures the main idea of
the analysis. We then briefly sketch the slight modification needed to eliminate this assumption.

First observe that if RSD is run on the agents in $G$ alone, then for a uniformly at
random ordering of agents in $G$, in expectation,
$\Omega(G)$ agents match to items along edges of utility $\Omega(1)$.
Now when RSD is run on the entire set of agents,
if the total expected utility achieved by RSD from agents in $G$ is at least
$\opt/(2n^{1/3})$, then we are already done.
Otherwise, there must exist a set $B_1$ of {\em bad} agents (need not be disjoint from $G$) such that (a) the list of each agent
in $B_1$ contains at least $2n^{1/3}$ elements, and (b) $|B_1| \ge 2n^{1/3}|G|$. Let

$$I_1 = \bigcup_{i \in B_1} \bigcup_{1 \le j \le n^{1/3}} L_i(j)$$

Observe that for {\em any} ordering of agents in $B_1$, the first $n^{1/3}$ agents in $B_1$ are able to match to an item in the first-half of their respective lists, deriving a utility of at least $n^{1/3}/2$.
Thus in absence of any other agents, any execution of RSD on agents in $B_1$ alone derives a total utility of $n^{1/3}/2$ or more.
Now when RSD is run on the entire set of agents,
if the total expected utility in RSD from agents in $B_1$ is at least $(n^{1/3})/2$,
we are once again done since $\opt \le n^{2/3}$ by our earlier assumption.

If not, there must exist a set $B_2$ of {\em very bad} agents
(need not be disjoint from $G \cup B_1$) such that

$$  \left( \frac{|B_1|}{|B_2|} \right) \left( \frac{n^{1/3}}{2} \right) < \frac{\opt}{2n^{1/3}}.$$

Thus

$$ |B_2| > \frac{n^{2/3} |B_1|}{\opt} \ge 2n,$$

a contradiction!

Finally, to eliminate the assumption that $\opt$ matches agents in $G$ to the set $I_G$
along edges of utility $\Omega(1)$, we note that there must exist a set $G' \subseteq G$
and a set $I' \subseteq I_G$ such that for some $\alpha \in (0,1]$,
(i) $\opt$ matches items in $I'$ to agents in $I_G$
along edges of utility $\Theta(\alpha)$, and (ii) the utility contributed by this matching is $\Omega(\opt/\log n)$. This follows from a standard geometric grouping argument. We can now focus
on agents-item pairs defined by $G'$ and $I'$, and apply the same argument as above.

\medskip
\noindent
{\bf Without Randomness:}
We note that randomness is crucial to achieving non-trivial linear welfare factor
in case of partial preference lists. If the permutation
$\sigma$ on agents is chosen deterministically, then the linear welfare factor drops to
$\Theta(\log n/n)$. We refer to this algorithm as the {\em serial dictatorship} (SD) algorithm.

An $O(\log n/n)$ bound on the linear welfare factor of SD follows from the following
family of instances where $\opt$ is $n/2$ while there exists an arrival
order on which SD achieves only $O(\log n)$ overall utility.
Assume that $n$ is an even integer.
There are $n/2$ items in the instance, numbered $1$ through $n/2$.
The ordered list $L_i$ is defined
to be $\{ 1, 2, ..., i \}$, for $1 \le i \le n/2$, and $L_i = \{ i - (n/2) \}$, for $n/2 < i \le n$.
Suppose the agents arrive in order $1$ through $n$. Then it is easy to see that for each
$i \in \{ 1, 2, ..., n/2 \}$, agent $i$ chooses item $i$ and obtains utility of $1/i$.
Thus the total utility achieved is $H(n/2) = \Theta(\log n)$. On the other hand,
the optimal allocation assigns item $i$ to agent $i + (n/2)$, for $1 \le i \le n/2$,
achieving an overall utility of $n/2$.

On the other hand, SD always achieves a linear welfare factor of $\Omega(\log n/n)$.
Fix an optimal solution, say $\opt$, and let $k$ be the number of items
allocated in $\opt$. Clearly, $k \ge |\opt|$. Since SD always outputs a maximal matching
between agents and items, it follows that the total number of items allocated in any execution of
SD must be at least $k/2$. Fix an execution of SD, and let $I'$ be the ordered set of items that get
allocated in SD. Then it is easy to see that when the $j^{\rm th}$ item in $I'$
is allocated, the utility derived is at least $1/j$. Thus the total utility obtained by SD is at least
$H(k/2) = \Omega(\log k)$. The performance ratio of SD is at most
$\max_{1 \le k \le n} k/H(k/2) = \Theta(n/\log n)$.

\subsection{Non-Unit Demands}

We now analyze the ordinal welfare factor of RSD when each agent desires a bundle of $K$
items for some integer parameter $K$ and each agent has a preference list that 
defines an ordering over the subsets of items of size $K$.
We show that in this case, the ordinal welfare factor of RSD is $\Theta(1/K)$. We will show that
an upper bound of $1/K$ holds on the RSD ordinal welfare factor even when preference lists are complete
while the RSD ordinal welfare factor is at least $\Omega(1/K)$ if the preference lists are partial.

To see that the ordinal welfare factor is bounded from above by $1/K$, consider the following instance
with $n$ agents and $nK$ items where the agents are partitioned into $p = n/K$ groups
of $K$ agents each. Let $A_1, A_2, ..., A_p$ denote these groups of agents.
For each $1 \le j \le K$ and $1 \le i \le p$,
the preference list of the $j^{th}$ agent in group $A_i$ is as follows:
the first choice is the item set $\{ (i-1)K^2+1, (i-1)K^2+K+1, ..., (i-1)K^2+ (K-1)K +1 \}$,
the second choice is the item set $\{ (i-1)K^2 + (j-1)K +1, (i-1)K^2 + (j-1)K +2,
..., (i-1)K^2 + (j-1)K +K \}$, the next $({K^2 \choose K} -2)$ choices consist of an
arbitrary ordering of the remaining $K$-element subsets of items $\{ (i-1)K^2+1, (i-1)K^2+2, ..., (i-1)K^2 + K^2 \}$,
and finally, this is followed by an arbitrary ordering on the remaining subsets.
Thus the first choice of each agent in any group $A_i$ intersects
with the first and second choices of every other agent in the group. Moreover, it is easy to see that
each agent also matches to one of the first ${K^2 \choose K}$ choices on its list. Thus allocations
made to agents in one group never interfere with the allocations made to agents in another group.
Now fix the matching $M$ that allocates to each agent its second preference -- it is easy
to verify that this allocation is conflict-free. Consider now a run of the RSD algorithm.
The first agent to arrive in a group $A_i$ always gets his first preference. However,
once this allocation is done, no other agent in the group $A_i$ can be allocated any item.
Thus the ordinal welfare factor of RSD on this instance is precisely $1/K$. \smallskip

On the other hand, an argument similar to that in Theorem \ref{thm:rsd-oe}
shows that the performance of RSD is at least $\frac{1}{2K} - o(\frac{1}{K})$.
The argument is the following: by the time the $t^{th}$ agent arrives, at most $Kt$ items have been taken away.
Of the $(n - t)$ agents remaining, at least $n - (K+1)t$
of them have their optimal bundle intact, as $Kt$ items can hit at most the bundles of $Kt$ agents.
So, the probability the $t^{th}$ agent is happy is at least $\frac{n - (K+1)t}{n-t}$.  Therefore, the expected number
of happy agents is at least

$$\sum_{t=1}^{n/(K+1)} (n - (K+1)t)/(n - t) $$
\noindent
The expression simplifies to $$n  - Kn\cdot\sum_{t=1}^{n/(K+1)}\frac{1}{n-t} \approx n\left(1 - K(\ln(1+1/K))\right) \ge n/2K - o(1/K)$$

\section{Conclusions and Open Directions}
In this paper, we studied the social welfare of two well studied mechanisms, RSD and PS,
for one-sided matching markets. We focussed on two measures: one was the ordinal welfare
factor which we believe is a new way to measure efficiency when agents utilities are
obtained as rankings or preference lists, the other was the natural linear utilities measure
which assumed that the utility of an agent dropped linearly down his preference list.

We performed a tight analysis of the ordinal welfare factors of both mechanisms, and the
linear welfare factor in the case of efficient instances: instances where the optimum matching
gave every agent one of his first $o(n)$ top choices. An open problem is to perform a tight
analysis of the case of linear welfare factor for general instances.%; we believe this will require a
%new idea, but should be possible.

We think the notion of ordinal welfare factor and the definitions of truthfulness in case
of ordinal utilities will be useful for other problems as well where the utilities are
expressed as preference lists rather than precise numbers. Examples which come to
mind are scheduling, voting, and ranking applications.

\bibliographystyle{plain}
\bibliography{huge}

\appendix
\section{Comparison of PS and RSD}\label{sec:psvsrsd}
A question might arise whether one of the two algorithms is as good as the other in all instances, when one considers the two efficiency measures studied in the paper.
% of ordinal welfare factor and linear welfare factor.
In this subsection we show, via examples, that such a statement doesn't hold true for either measure.
Both examples are taken from the paper of Bogomolnaia and Moulin~\cite{bogo-moulin}.\smallskip

\noindent
{\em Instance 1}
\begin{table}[ht]
\begin{minipage}[b]{0.3\linewidth}\centering
\begin{tabular}{ccccc}
\\
1 & $a$ & $b$ & $c$ & $d$\\
2 & $a$ & $b$ & $c$ & $d$\\
3 & $b$ & $a$ & $d$ & $c$\\
4 & $b$ & $a$ & $d$ & $c$
%\hline
\end{tabular}
\end{minipage}
\hspace{0.4cm}
\begin{minipage}[b]{0.3\linewidth}
\centering
\begin{tabular}{ccccc}
   & $a$ & $b$ & $c$ & $d$\\
1 & $5/12$ & $1/12$ & $5/12$ & $1/12$\\
2 & $5/12$ & $1/12$ & $5/12$ & $1/12$\\
3 & $1/12$ & $5/12$ & $1/12$ & $5/12$\\
4 & $1/12$ & $5/12$ & $1/12$ & $5/12$
\end{tabular}
\end{minipage}
\hspace{0.4cm}
\begin{minipage}[b]{0.3\linewidth}
\centering
\begin{tabular}{ccccc}
   & $a$ & $b$ & $c$ & $d$\\
1 & $1/2$ & $0$ & $1/2$ & $0$\\
2 & $1/2$ & $0$ & $1/2$ & $1/12$\\
3 & $0$ & $1/2$ & $0$ & $1/2$\\
4 & $0$ & $1/2$ & $0$ & $1/2$
\end{tabular}
\end{minipage}
\end{table}

\noindent
{\em Instance 2}
\begin{table}[ht]
\begin{minipage}[b]{0.3\linewidth}\centering
\begin{tabular}{cccc}
\\
1 & $a$ & $b$ & $c$ \\
2 & $a$ & $c$ & $b$ \\
3 & $b$ & $a$ & $c$
%\hline
\end{tabular}
\end{minipage}
\hspace{0.4cm}
\begin{minipage}[b]{0.3\linewidth}
\centering
\begin{tabular}{cccc}
   & $a$ & $b$ & $c$\\
1 & $1/2$ & $1/6$ & $1/3$ \\
2 & $1/2$ & $0$ & $1/2$ \\
3 & $0$ & $5/6$ & $1/6$
\end{tabular}
\end{minipage}
\hspace{0.4cm}
\begin{minipage}[b]{0.3\linewidth}
\centering
\begin{tabular}{cccc}
   & $a$ & $b$ & $c$\\
1 & $1/2$ & $1/4$ & $1/4$ \\
2 & $1/2$ & $0$ & $1/2$ \\
3 & $0$ & $3/4$ & $1/4$
\end{tabular}
\end{minipage}
\end{table}

\noindent
In the two instances above, the first table shows the preference lists of agents (numbered $1$,$2,\ldots$)
over items (denoted as $a,b,\ldots$). The second table shows the randomized allocation obtained form the RSD mechanism, the third shows that obtained via the PS mechanism.

Let us consider the linear utilities case first. In the first instance, RSD gets a utility of $4\cdot\frac{1}{4}\cdot\left(4\cdot5/12 + 3\cdot 1/12 + 2 \cdot 5/12 + 1\cdot 1/12\right) = 34/12$. While PS obtains a utility of
$4\cdot\frac{1}{4}\cdot\left(4\cdot 1/2 + 2\cdot 1/2\right) = 3$. So PS does better than RSD. In the second instance, the calculation shows the utility of RSD is greater than that of PS: the utilities on the allocations where they differ for RSD is $\frac{1}{3}\cdot(2\cdot1/6 + 1\cdot1/3 + 3\cdot5/6 + 1\cdot1/6) = 10/9$, while for PS is $\frac{1}{3}\cdot(2\cdot1/4 + 1\cdot1/4 + 3\cdot3/4 + 1\cdot1/4) = 13/12$.

Moving to ordinal welfare factor, in the second instance, if the benchmark matching is $(1,a),(2,c),(3,b)$, then the expected number of agents made happy by PS is $9/4$, while that in RSD is $7/3$. However if the benchmark matching is $(1,b),(2,a),(3,c)$, then the ordinal welfare factors of PS and RSD are $9/4$ and $13/6$ respectively.  It is easy to check that both benchmark matchings are pareto-optimal.

\section{Performance of RSD with linear utilities}\label{app:rsd}
Let $\alpha$ and $\beta$ be two parameters between $0$ and $1$ to be fixed later.
We assume that at least $\alpha n$ agents get an item which is in their top
$\beta n$ choices. It is easy to see that if not, the optimum is at most
$(1 - \beta + \alpha\beta)n$ and RSD gives a linear welfare factor of $\frac{1}{2(1-\beta + \alpha\beta)}$.

Call these $\alpha n$ agents good. Now consider the argument in
Lemma \ref{thm:rsd-lin}. At time $t$, the expected number of remaining good agents
is $\alpha(n-t)$. Thus, with probability at least $\left(\alpha - \frac{d_t}{n-t}\right)$, we choose a good agent and he gets an item from his top $\beta n$ choices, thus he gets a utility of at least $1 - \beta$, and so
the benefit over the `trivial' $(1-t/n)$ utility is $(t/n - \beta)$.
Thus, the total benefit obtained at time $t$ is at least $\left( \alpha - \frac{d_t}{n-t} \right)\cdot(t/n - \beta)$.
Using the fact that $d_t \le (t+2)(n-t)/(n+1)$, we get that the total benefit
over $50\%$ is at least
$ \int_{t_1}^{t_2} \left(\alpha - \frac{t+2}{n+1}\right)\left(\frac{t}{n} - \beta \right)$, where
$t_1 = \beta n$ (when the second term becomes positive), and $t_2 = \alpha*(n+1) - 2$ (after which the first term becomes negative).
% \ge n\left(\alpha/2 + \beta/2 - \alpha\beta - 1/3 \right) $
%\\
%\noindent
%$n(\frac{\alpha}{2}(1 - 1/e)^2 - \frac{3e^2 - 8e + 3}{4e^2} - \alpha\beta(1 - 1/e) + \beta(1 - 2/e))$
Therefore, the RSD mechanism gives an approximation of at least the minimum of this integral and
$\frac{1}{2(1 - \beta + \alpha\beta)}$, for any $\alpha$ and $\beta$.
For $\alpha = 0.77$ and $\beta = 0.22$, we get the above expression to be $52.6\%$.

\section{Performance of PS with linear utilities}\label{app:ps}
%\begin{theorem}\label{thm:pslinear}
%With linear utilities, the expected social welfare of the PS mechanism is at least $0.6602$ of the social optimum.
%\end{theorem}
%\begin{proof}
We now consider the performance of PS with linear utilities even when the instance may not be efficient, i.e. some agents get less than $1-o(1)$ in the optimum matching. Let $M^*$ be the optimal matching.

Consider items in the order in which they are fully allocated in the run of PS algorithm (break ties arbitrarily).
For an item $i$ in this order, let $a$ be the agent such that $M*(a)=i$ and let $u_i$ be its utility for item $i$. The social optimum is $\sum_i u_i$. We denote it by $\opt$.
%What is the utility obtained by this agent $a$ in PS?
Now we compute the total utility of items assigned to $a$ in PS. From Observation \ref{obs:1},
we get that in the time interval $[\frac{k-1}{n},\frac{k}{n}]$ for $1\le k\le i$, agent $a$ will obtain utility at
a rate of at least $\mbox{max}\{u_i, (n-k+1)/n\}$. This is because in this time interval both $M^*(a)$ and his $k^{th}$ choice are available. Similarly, for $k>i$, agent $a$ gets utility at a rate $(n-k+1)/n$
in the interval $[\frac{k-1}{n},\frac{k}{n}]$.
So, the total utility of agent $a$ is at least $\sum_{k=1}^n \frac{n - k+1}{n} \ge 1/2$ if $u_i \le (n-i)/n$, and it is
$\frac{1}{2}+\frac{\left(u_i-\frac{n-i}{n}\right)^2}{2}$ otherwise. We denote the total utility of the assignment in the PS algorithm by $\alg$.

Observe that, if for any $i$, $u_i<(n-i)/n$, then increasing $u_i$ increases the $\opt$  but leaves $\alg$ unchanged. Since we are looking at an upper bound on the performance PS algorithm, it  is safe to assume that, for every $i$, $u_i \ge (n-i)/n$.
We want to show that there exists a constant $\beta\le 1$ such that
\begin{align}
\alg\ge \beta \opt\label{pslinear}
\end{align}
for any set of values of $u_i,0\le u_i\le 1$.
By taking a partial derivative w.r.t. $u_i$, and using the fact that $u_i \le 1$, we see that the minimizer of $\alg$ is
$\mbox{min}\left(1, \frac{n-i}{n}+\beta\right)$.
%$$u_i=\frac{n-i}{n}+\beta$$
%Observe that $\alg-\beta \sum_i opt_i$ has degree $2$ and its value at $u_i=\pm\infty$ is $\infty$. Thus  at $u_i=\frac{n-i}{n}+beta$, it reaches its minima. Since $u_i$ cannot be greater than $1$, the value of $u_i$ that minimizes the above equation is $\mbox{min}\left(1, \frac{n-i}{n}+\beta\right)$.
Now we compute the values of $\opt$ and $\alg$ with these values of $u_i$.
\begin{align}
\opt &=\sum_i(u_i) \\
&=\sum_{1\le i\le \beta n} 1 + \sum_{\beta n\le i\le n}\left(\beta-\frac{i}{n}\right) \\
&=\frac{n}{2}+\beta n-\frac{\beta^2}{2}
\end{align}

\begin{align}
\alg &=\frac{n}{2}+ \sum_{1\le i\le n} \frac{\left(u_i-\frac{n-i}{n}\right)^2}{2} \\
&=\frac{n}{2}+\sum_{1\le i\le \beta n} \frac{i^2}{2}+\sum_{\beta n\le i\le n} \frac{\beta^2n}{2}\\
&=\frac{n}{2}+\frac{\beta n}{2}(1-\beta) +\frac{\beta^3}{6}
\end{align}
Recall that $\beta$ should satisfy $\alg \ge \beta \opt$. Solving for this gives $\beta = 0.6602$.
%This, along with Lemma \ref{thm:ps-lin-lb} completes the proof of (b) of Theorem \ref{thm:ps}.
It will be interesting to close the curiously small gap between the upper bound and the lower bound.
%\end{proof}

\end{document}